\documentclass[journal]{IEEEtran}
\usepackage{tipa}
\usepackage{amsfonts}
\usepackage{amsfonts}
\usepackage{amsfonts}
\usepackage{mathbbold}
\usepackage{amsfonts}
\usepackage{amssymb}
\usepackage{stfloats}
\usepackage{cite}
\usepackage{graphicx}
\usepackage{psfrag}
\usepackage{subfigure}
\usepackage{amsmath}
\usepackage{array}

\newcommand{\ThetaL}{\mathbf{\Theta}}
\newcommand{\limK}{\lim\limits_{K\rightarrow\infty}}

\begin{document}

\title{Opportunistic Buffered Decode-Wait-and-Forward (OBDWF) Protocol for Mobile Wireless Relay Networks}

\newtheorem{Problem}{Problem}
\newtheorem{Algorithm}{Algorithm}
\newtheorem{Subproblem}{Subproblem}
\newtheorem{Protocol}{Protocol}
\newtheorem{Proposition}{Proposition}
\newtheorem{Lemma}{Lemma}
\newtheorem{Theorem}{Theorem}
\newtheorem{Def}{Definition}
\newtheorem{Corollary}{Corollary}
\newtheorem{Scheme}{Scheme}
\newtheorem{Baseline}{Baseline}
\newtheorem{Rem}{Remark}
\newtheorem{Assumption}{Assumption}

\author{ \authorblockN{Rui Wang,  Vincent K. N. Lau, {\em Senior Member, IEEE}, and Huang Huang, {\em Student Member, IEEE}}
\thanks{The authors are with
the Department of Electronic and Computer Engineering (ECE), The
Hong Kong University of Science and Technology (HKUST), Hong Kong.
(email:  ray.wang.rui@gmail.com, eeknlau@ee.ust.hk,huang@ust.hk).
This work was supported by the Research Grants Council of the Hong
Kong Government through the grant RGC 615609.} }

\markboth{To be appeared in IEEE Trans. Wireless Commun.} {Shell
\MakeLowercase{\textit{et al.}}: Bare Demo of IEEEtran.cls for
Journals}

\maketitle

\begin{abstract}
In this paper, we propose an opportunistic buffered
decode-wait-and-forward (OBDWF) protocol to exploit both {\em relay
buffering} and {\em relay mobility} to enhance the system throughput
and the end-to-end packet delay under bursty arrivals. We consider a
point-to-point communication link assisted by $K$ mobile relays. We
illustrate that the OBDWF protocol could achieve a better throughput
and delay performance compared with existing baseline systems such
as the conventional dynamic decode-and-forward (DDF) and
amplified-and-forward (AF) protocol. In addition to simulation
performance, we also derived closed-form asymptotic throughput and
delay expressions of the OBDWF protocol. Specifically, the proposed
OBDWF protocol achieves an asymptotic throughput $\ThetaL(\log_2 K)$
with $\ThetaL(1)$ total transmit power in the relay network. This is
a significant gain compared with the best known performance in
conventional protocols ($\ThetaL(\log_2 K)$ throughput with
$\ThetaL(K)$ total transmit power). With bursty arrivals, we show
that both the stability region and average delay of the proposed
OBDWF protocol can achieve order-wise performance gain $\ThetaL(K)$
compared with conventional DDF protocol.

\end{abstract}


\section{Introduction}\label{sec:intro}

In wireless communication networks, cooperative relaying not only
extends the coverage but also contributes the \textit{spatial
diversity}. As a result, cooperative relay is one of the core
technology components in the next generation wireless systems such
as IEEE 802.16m and LTE-A. There are extensive studies on the theory
and algorithm design of cooperative relay, { and they can be roughly
classified as focusing on the {\em microscopic} and {\em
macroscopic} aspects.} Examples of microscopic studies include the
decode-and-forward (DF), amplify-and-forward (AF) and
compress-and-forward (CF) protocols for single hop
\cite{Kramer:05,Relay:Tse:2004,DDF:2003} as well as multi-hop
cooperative relay networks \cite{relay:book:2006}. In
\cite{AF:DoF:2009}, the authors demonstrated that AF could achieve
optimal end-to-end DoF for the MIMO point-to-point system with
multiple relay nodes. In \cite{CF:Cap:2009}, it is shown that a
variant of the CF relaying achieves the capacity of any general
multi-antenna Gaussian relay network within a constant number of
bit. In \cite{Relay:Sel:08}, relay selection protocol is shown to
achieve higher bandwidth efficiency while guaranteeing the same
diversity order as that of the conventional cooperative protocols.
However, all these works have focused on the physical layer
performance (such as throughput) and failed to exploit the buffer
dynamics in the relay. Furthermore, they have assumed all the relays
are static and have ignored the potential benefit introduced by
mobility in the network.  On the other hand, there are also some
papers focusing on studying the macroscopic behavior of cooperative
ad-hoc networks. For example, the scaling law of the wireless ad-hoc
network is derived in
\cite{Gupta:00,GamalAE:04,Toumpis:04,Kulkarni:04} and it is shown
that each node can achieve the throughput of the order
$O(\frac{1}{\sqrt{K\log_2 K}})$\footnote{$f(K)=O(g(K))$ means that
there exists a constant $C$ such that $f(K)\leq Cg(K)$ for
sufficiently large $K$, $f(K)=o(g(K))$ means that $\limK
\frac{f(K)}{g(K)}=0$, and $f(K)=\ThetaL(g(K))$ means that
$f(K)=O(g(K))$ and $g(K)=O(f(K))$.} when $K$ fixed nodes are
randomly distributed over a unit area. These results imply that the
throughput of each node converges to zero when the number of nodes
increases. Nevertheless, it is found in \cite{Ozgur:07} that the
per-node throughput can arbitrarily close to constant by
hierarchical cooperation. In \cite{Gastpar:02b,Gastpar:05}, it's
shown that the source-destination throughput can scale as
$\ThetaL(\log_2 K)$ when all the relays amplify and forward the
received packet to the destination cooperatively with $\ThetaL(K)$
total transmission power. In \cite{Grossglauser:02}, the authors
have shown that a per link throughput of $\ThetaL(1)$ can be
achieved at the expense of potentially large delay when the nodes
are mobile. All these works have suggested that there are potential
advantage of relay buffering and relay mobility. However, there are
also various technical challenges to be addressed before we could
better understand the potential benefits.
\begin{itemize}
\item {\bf Low Complexity Relay Protocol Design Exploiting Relay Buffering and
Relay Mobility:} Although the idea of utilizing the mobility has
been studied in \cite{Grossglauser:02,Thrput-delay:ad-hoc} in the
study of ad-hoc network throughput analysis (scaling laws), there is
not much work that addresses the microscopic details (such as
protocol design) of the problem. For example, most of the existing
relay protocols have focused entirely on the physical layer
performance (information theoretical capacity or Degrees-of-Freedom
(DoF)) and they did not fully exploit the potential of relay
buffering. In fact, it is quite challenging to design low complexity
relay protocol that could exploit both the relay buffering and relay
mobility. Furthermore, the issue is further complicated by the
bursty source arrival and randomly coupled queue dynamics in the
systems.

\item {\bf Performance Analysis:} It is very important to have closed form
performance analysis to obtain insights to understand the
fundamental tradeoff between throughput gain and delay penalty in
cooperative systems. However, it is very challenging to analyze
closed form tradeoff between the throughput, stability region and
end-to-end delay. For instance, most of the existing papers studying
delay and throughput scaling laws in ad-hoc network
\cite{Grossglauser:02,Thrput-delay:ad-hoc} are focused on the
macroscopic aspects of the systems and they have ignored the
microscopic details such as the random bursty arrivals and queue
dynamics in the systems. When these dynamics are taken into
consideration, the problem involves both information theory (to
model the physical layer dynamics) and the queueing theory (to model
the queue dynamics), which is highly non-trivial.

\end{itemize}

In this paper, we shall propose an opportunistic buffered
decode-wait-and-forward (OBDWF) protocol to exploit both {\em relay
buffering} and {\em relay mobility} to enhance the system throughput
and the end-to-end packet delay under bursty arrivals. We consider a
point-to-point communication link assisted by $K$ mobile relays. By
exploiting the {\em relay buffering} and {\em relay mobility} in the
phase I and phase II of the proposed OBDWF, we shall illustrate that
the OBDWF could achieve a better throughput and delay performance
compared with existing baseline systems such as the conventional
Dynamic Decode-and-Forward (DDF) \cite{DDF:2003} and AF protocol
\cite{AF:sel}. In addition to simulation performance, we also
derived closed-form asymptotic throughput and delay expressions of
the OBDWF protocol. { Under random bursty arrivals and queue
dynamics, the proposed OBDWF protocol has low complexity with only
$O(K)$ communication overhead and $\ThetaL(1)$ total transmit power
in the relay network. It achieves a throughput $\ThetaL(\log_2 K)$,
which is a significant gain compared with the best known performance
in conventional protocols ($\ThetaL(\log_2 K)$ throughput with
$\ThetaL(K)$ total transmit power). Furthermore, both the stability
region and average delay can achieve order-wise performance gain
$\ThetaL(K)$ compared with conventional DDF protocol.}


\section{System Model}
\label{sec:model}

\begin{figure}
 \begin{center}
  \resizebox{8cm}{!}{\includegraphics{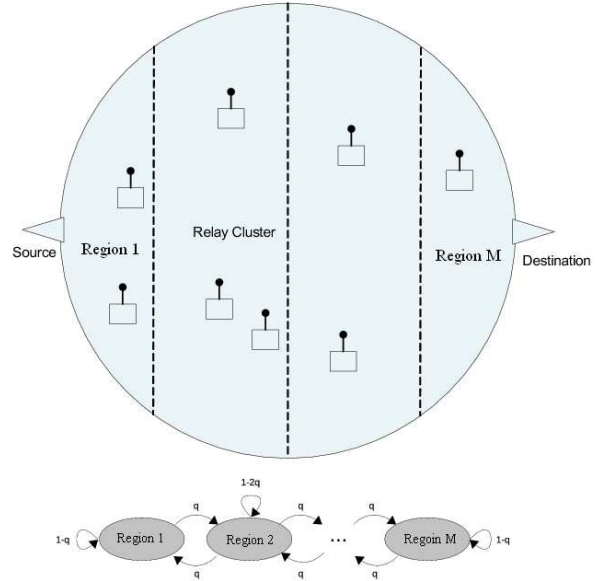}}
 \end{center}
    \caption{ System and random walk mobility model for the point-to-point communication
system with a half-duplexing mobile relay network. The parameter $q$
determines the mobility of the network (larger $q$ means higher
mobility).}
    \label{fig:system_model}
\end{figure}

In this section, we shall describe the system model for the
point-to-point communication system with a half-duplexing mobile
relay network. Specifically, there are $K$ mobile relay nodes
between the source node and the destination node, as shown in Fig.
\ref{fig:system_model}. 
Each of the source node and the
$K$ relay nodes has an infinite length buffer. To facilitate the
relay scheduling, transmission is partitioned into frames. Each
frame is further divided into three
types of slots defined as follows: 
\begin{itemize}
\item {\bf Channel Estimation Slot} is used by relays for
estimating the channel gains with the source and destination nodes.

\item {\bf Control Slot} is used by relays for distributive
control signaling of the OBDWF protocol. The details is given in the
protocol description.

\item {\bf Transmission Slot} is used for data transmission, and it last $\tau$ seconds.
\end{itemize}

\subsection{Relay Mobility Model}

Following \cite{Gupta:00,Gastpar:05}, we assume that the $K$ relays
are distributed on a disk with radius $R$ as illustrated in Fig.
\ref{fig:system_model}. The source and the destination nodes are
fixed at two ends of a diameter, and the disk is divided
horizontally into $M$ equal-area regions along the
source-destination diameter. These regions are denoted as region 1,
region 2, ..., and region $M$, from the source to the destination.
As illustrated in Fig. \ref{fig:system_model}, the movement of each
relay is modeled as a random walk (Markov chain) over these regions:

\begin{itemize}
\item At the beginning, each relay is uniformly distributed on the
disk. Movements of relays can only occur in discrete frames with
time index $t$.

\item Let $X_k(t)$ denote the region index of the $k$-th relay in the
$t$-th frame, $\{X_k(t)|t=1,...+\infty\}$ is a Markov chain with the
following transition matrix
\begin{equation}\label{eq:random_walk}
\begin{array}{l}
\Pr\{X_k(t+1)|X_k(t)\}=\\
\left\{ \begin{array}{ll}
q & X_k(t+1)=X_k(t)\pm1 \\
1-2q & X_k(t+1)=X_k(t)=2,...,M-1\\
1-q & X_k(t+1)=X_k(t)=1 \mbox{ or } M \\
0 & \mbox{ otherwise }
\end{array} \right.
\end{array}
\end{equation}
\item When one relay moves into a region, its actual location in
this region is uniformly distributed.
\end{itemize}

\begin{Rem}[Interpretation of Parameter $q$] The region transition probability $q\in(0,\frac{1}{2}]$ measures how
likely one relay will move into another region in the next frame,
and therefore, it is related to the average speed of the relays.
~\hfill \IEEEQED
\end{Rem}

\subsection{Physical Layer Model}
Let $H_{s,j}$ and $d_{s,j}$ be the small scale fading gain and the
distance between the source node and the $j$-th relay respectively,
and let $H_{j,d}$ and $d_{j,d}$ be the small scale fading gain and
the distance between the $j$-th relay and destination node
respectively.
\begin{Assumption}[Assumption on the Channel Model]
We assume that $\{H_{s,j},H_{j,d}\}$ are quasi static in each frame.
Furthermore, $\{H_{s,j},H_{j,d}\}$ are i.i.d over frames according
to a general distribution $\Pr\{H\}$ and independent between each
link.   ~\hfill \IEEEQED
\end{Assumption}

The relay network shares a common spectrum with bandwidth $W$Hz, and
each node transmits at a peak power $P$. Let $x_s$ be the
transmitted symbol from the source node, the received signal at the
$j$-th relay is given by: $y_j =
\sqrt{H_{s,j}/{d_{s,j}^{\alpha}}}x_s + z_j$, where $\alpha$ is the
path loss exponent, and $z_j$ is the i.i.d $\mathcal{N}(0,1)$ noise.
The achievable data rate between the source node and the $j$-th
relay is given by:
\begin{equation}
C_{s,j} = W\log_2 (1+P\xi H_{s,j}/{d_{s,j}^{\alpha}})
\label{eqn:cp-sj} ,
\end{equation}
where $\xi\in(0,1]$ is a constant can be used to model both the
coded and uncoded systems. 
Similarly, the achievable data rate between $j$-th relay and the
destination node is given by
\begin{equation}
C_{j,d} = W\log_2 (1+P\xi
H_{j,d}/{d_{j,d}^{\alpha}})\label{eqn:cp-jd} ,
\end{equation}

All the packets are transmitted at data rate $R=W\log_2\beta$ for
some constant $\beta$. The $j$-th relay could correctly decode the
packets transmitted from the source node if $R\leq C_{s,j}$, and the
destination node could correctly decode the packets transmitted from
the $j$-th relay if $R\leq C_{j,d}$. For easy discussion, we shall
denote a link as a {\it connected link} if its achievable data rate
is larger than $R$, and otherwise a {\it broken link}.

\section{The OBDWF Protocol}
\label{sec:protocol} In this section, we shall first describe the
proposed opportunistic buffered decode-wait-and-forward (OBDWF)
relay protocol for mobile relays.
\begin{Protocol}[OBDWF Protocol for Mobile Relays]
\
\begin{itemize}
\item[1.] Each relay measures the current states \{{\em connected, broken}\} of its links with the
destination node in the channel estimation slot.


\item[2.] The control slot is divided into mini-slots.
If the buffer in a relay is not empty and the link state to the
destination is {\em connected}, it will submit a request in a
control mini-slot. Using standard contention mechanism, one active
relay is selected to transmit its packet from its buffer to the
destination node\footnote{ The algorithm can be extended to consider
spatial combining from multiple relays. However, the performance
gain associated with that will be quite limited due to the path loss
effects. For instance, there is very low chance of having multiple
relays near the source or multiple relays (having the same common
packet) near the destination.}. The selected relay as well as all
the other relays will dequeue the same packet from their buffers.

\item[3.]
If none of the $K$ relays attempts to compete for access to the
destination node in the control slot, the source node will broadcast
a new packet to the $K$ relays { and the destination node} using a
fixed data rate $R=W\log_2 \beta$ for some constant $\beta$. The
source node will dequeue the packet from its buffer if there is an
ACK from at least one of the relays or { the destination node}.
~\hfill \IEEEQED
\end{itemize}\label{sch:macro-op}
\end{Protocol}

{ Note that the OBDWF protocol has only $O(K)$ communication
overhead with $\ThetaL(1)$ total transmit power in the relay
network, which is the same as the conventional DDF and conventional
AF protocol, elaborated in Table \ref{tab:baselines}. Unlike
conventional DDF protocol where the phase I (source to relay) and
the phase II (relay to destination) of the same packet appear as
inseparable atomic actions, the proposed OBDWF protocol exploits
buffers in the relays to create the flexibility to schedule phase I
and phase II of the same packet based on the instantaneous channel
state as illustrated in Fig. \ref{fig:timing}. Coupled with relay
mobility, the proposed OBDWF protocol allows the relay to buffer the
packet and wait for good opportunity (when the relays is close to
the destination) to deliver the packet. As a result, relay mobility
allows the system to operate at a higher throughput at the expense
of larger delay. We shall quantify such tradeoff in Section
\ref{sec:thrput} and \ref{sec:delay}. }

\begin{figure}
 \begin{center}
  \resizebox{9cm}{!}{\includegraphics{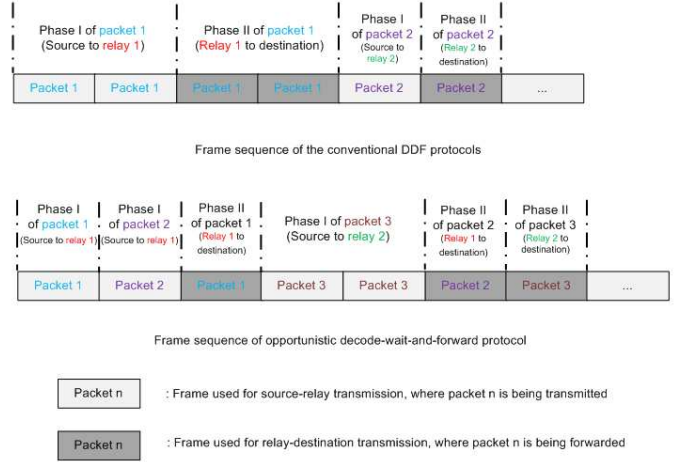}}
 \end{center}
    \caption{ Frame sequences for conventional DDF protocl and opportunistic buffered decode-wait-and-forward (OBDWF) protocol.
In conventional DDF protocol, the phase I and the phase II of the
same packet appear as inseparable atomic actions. On the other hand,
the proposed OBDWF protocol exploits buffers in the relays to create
the flexibility to schedule phase I and phase II of the same packet
based on the instantaneous channel state. Coupled with relay
mobility, the proposed OBDWF protocol allows the relay to buffer the
packet and wait for good opportunity (when the relays is close to
the destination) to deliver the packet.
    }
    \label{fig:timing}
\end{figure}


\section{Throughput Performance with Infinite Backlog} \label{sec:rlt-macro}
\label{sec:thrput} In this section, we shall first discuss the
average system throughput of the proposed OBDWF protocol with
infinite backlog at the source buffer. We first define the average
throughput below.

\begin{Def}[Average End-to-End System Throughput]
\label{def:avg_thrput} Let $T_i$ be the number of information bits
successfully received by the destination node at the $i$-th frame.
The average end-to-end system throughput between the source node and
the destination node $\overline{T}$ is defined as  $\overline{T} =
\lim_{\kappa\rightarrow\infty}\frac{\sum_{i=1}^\kappa
T_i}{\kappa\tau}$. ~\hfill \IEEEQED
\end{Def}

\subsection{Throughput Performance of the OBDWF Protocol}
Note that for a fixed number of relay nodes $K$, when the source
node increases the data rate $R$, the associate radio coverage and
the number of relays who can decode the source packet becomes
smaller. On the other hand, for fixed $R$, the number of relays who
can decode the source packet increases as $K$ increases. We shall
quantify such scaling relationship in Lemma \ref{lem:relay_decode}.
\begin{Lemma}[Scaling Relationship of Connected Link]\label{lem:relay_decode} Denote the transmit data
rate $R=W\log_2\beta$ and $\gamma=\beta^{\frac{2}{\alpha}}$ for some
constant $\beta$. For sufficiently large $K$, the following
statements are true:
\begin{itemize}
\item{I.}
If $\limK \frac{\gamma}{K}=0$, the probability that there are
$\ThetaL(\frac{K}{\gamma})$ relays having connected links with the
source node (or the destination node) tends to $1$.
\item{II.}
If $\limK \frac{\gamma}{K}>0$, the probability that there are
$\ThetaL(1)$ relays having connected links with the source node (or
the destination node) is $\ThetaL(\frac{K}{\gamma})$. ~\hfill
\IEEEQED
\end{itemize}
\end{Lemma}
\begin{proof}
Please refer to Appendix A.
\end{proof}

Using Lemma \ref{lem:relay_decode}, we obtain the closed-form
asymptotic average system throughput under infinite backlog at the
source buffer.
\begin{Theorem}{\em (Throughput Performance of the OBDWF
Protocol):}\label{thm:avg_thrput} For sufficiently large $K$ and
infinite backlog at the source buffer, the maximal average system
throughput of the proposed OBDWF protocol under the random walk
mobility model in (\ref{eq:random_walk}) is given by
\begin{equation}
\overline{T}^*_{\text{OBDWF}}=\ThetaL(\log_2 K)
\end{equation}
This order-wise throughput is achieved when
$\gamma=\ThetaL(K^{\sigma})$ for some constant $\sigma\in(0,1]$.
Furthermore, $\overline{T}^*_{\text{OBDWF}}$ is upper-bounded by but
infinitely close to $ \frac{W\alpha}{4} \log_2 K$. ~\hfill\IEEEQED
\end{Theorem}

\begin{proof}
Please refer to Appendix B.
\end{proof}

\begin{Rem}[Interpretation of Theorem \ref{thm:avg_thrput}]
Since there are infinitely large buffers at the relay nodes and the
random-walk transition probability $q$ is positive, the average
system throughput is $R/2$ as long as there are always relays having
connected links to the source and destination node (which is
presented mathematically as $\limK \frac{\gamma}{K}=0$ by Lemma
\ref{lem:relay_decode}). ~\hfill \IEEEQED
\end{Rem}

\subsection{Comparison with the Conventional DDF Protocol}

Similarly, we shall summarize the closed-form asymptotic average
system throughput for the conventional DDF protocol ({ elaborated in
Table \ref{tab:baselines}}) below.
\begin{Lemma}{\em (Throughput Performance of Conventional DDF Protocol):}\label{lem:macro-base}
For sufficiently large $K$ and infinite backlog at the source
buffer, the maximal average system throughput of the conventional
DDF protocol under the random walk mobility model in
(\ref{eq:random_walk}) is given by:
\begin{equation}
\overline{T}^*_{\text{DDF}} = \min \left\{\ThetaL((Kq^{M-1})^{1/M}),
\ThetaL \left(\log_2 (Kq^{M-1}) \right)\right\}
\end{equation}
This order-wise throughput is achieved when
\begin{equation}
\gamma =\max \left\{\ThetaL(1), \ThetaL(\sqrt{Kq^{M-1}})\right\}
\end{equation}
\end{Lemma}
\begin{proof}
Please refer to the Appendix C.
\end{proof}

Therefore, we have the following corollary on the performance gain
of the OBDWF protocol:

\begin{Corollary}{\em (Comparison of the Average System
Throughput):} \label{cor:thrput} The throughput gain of the OBDWF
protocol is
\begin{equation}
\frac{\overline{T}^*_{\text{OBDWF}}}{\overline{T}^*_{\text{DDF}}}=\left\{
\begin{array}{ll}
\mathbf{\Theta}\left(\frac{\log_2 K}{\log_2(Kq^{M-1})}\right) &
\text{if } \limK Kq^{M-1}=\infty\\
\mathbf{\Theta}\left(\frac{\log_2 K}{(Kq^{M-1})^{1/M}}\right) &
\text{otherwise}
\end{array}
\right.
\end{equation}
\end{Corollary}

\begin{Rem}[Interpretation of Corollary \ref{cor:thrput}]
Note that when the system mobility is low ($q=(\frac{\log_2
K}{K})^{\frac{1}{M-1}}$), there is an order-wise throughput gain
achieved by the OBDWF protocol. ~\hfill \IEEEQED
\end{Rem}

\section{Stability and Delay Performance with Bursty Arrivals}\label{sec:delay}
In this section, we shall focus on the stability region and the
delay performance analysis of the proposed OBDWF protocol under
bursty packet arrivals. We shall first define the busty source
model, followed by the analysis on the stability region and average
delay performance.

\subsection{Bursty Source Model}
Let $A_s(t)$ represents the number of new packets arriving at the
source node at the beginning of the $t$-th frame.
\begin{Def}[Bursty Source
Model]\label{ass:source_model} We assume that the arrival process
$A_s(t)$ is i.i.d over the frame index $t$ according to a general
distribution $\Pr\{A_s\}$. Each new packet has a fixed number of
bits $\overline{N}$. The first and second order moments of the
arrival process are denoted by $\lambda_s=\mathbb{E}[A_s]$ and
$\lambda_s^{(2)}=\mathbb{E}[A_s^2]$ respectively. ~\hfill \IEEEQED
\end{Def}

Let $Q_s(t)$ be the number of packets in the source buffer at the
$t$-th frame. The dynamics of the source buffer state is given by:
\begin{equation}
Q_s(t+1) = A_s(t+1) + \left[ Q_s(t) - U_s(t) \right]^{+}
\end{equation}
where $U_s(t)\in\{0,1\}$ is the number of packets transmitted to the
relay network at $t$-th frame. Furthermore, let $Q_k(t)$ be the
number of packets in the the $k$-th relay node's buffer at the
$t$-th frame. The dynamics of the relay buffer state is given by:
\begin{equation}
Q_k(t+1) = A_k(t+1) + \left[ Q_k(t) - U_k(t) \right]^{+}
\end{equation}
where $A_k(t)\in\{0,1\}$ is the number of packets received by the
$k$-th relay node from the source node at the beginning of the
$t$-th frame, and $U_k(t)\in\{0,1\}$ is the number of packets
dequeued from the $k$-th relay node at the $t$-th frame.

%

\subsection{Stability Performance}
In this section, we shall derive the stability region of the OBDWF
protocol and the conventional DDF protocols. We first define the
notion of queue stability \cite{Stability:1999,Stability:2005}
below.
\begin{Def}[Stability of the Queueing System]\label{def:stable}
The queueing system is {\em stable}, if
\begin{equation}
\lim_{t\rightarrow\infty}\Pr\{Q(t)<x\}=F(x) \text{ and
}\lim_{x\rightarrow\infty}F(x)=1
\end{equation}
where $Q(t)$ is the queue state in the queueing system at the $t$-th
frame. ~\hfill\IEEEQED
\end{Def}

Using Definition \ref{def:stable}, we have the following Theorem for
the OBDWF protocol.
\begin{Theorem}[Stability Region of the OBDWF
Protocol]\label{thm:stability} For sufficiently large $K$, the
system of queues $\{Q_s(t),Q_1(t),\cdots,Q_K(t)\}$ under the
proposed OBDWF protocol are stable if and only if
$\lambda_s<\frac{1}{\zeta}$, where $\zeta$ is given by
\begin{equation}
\label{eq:zeta}
\zeta=\max\left\{\ThetaL(1),\ThetaL(\gamma/{K})\right\}
\end{equation}
\end{Theorem}
\begin{proof}
Please refer to Appendix E.
\end{proof}

Similarly, the stability region of the conventional DDF protocol is
given by:
\begin{Lemma}[Stability Region of the Conventional DDF
Protocol]\label{lem:stability_base} For sufficiently large $K$, the
system of queues $\{Q_s(t),Q_1(t),\cdots,Q_K(t)\}$ under
conventional DDF protocol are stable if and only if
$\lambda_s<\frac{1}{\overline{D}_S}$, where $\overline{D}_S$ is
given by
\begin{equation}
\label{eq:DR_maintext} \begin{array}{l}\overline{D}_S =\\
 \left\{
\begin{array}{ll}
\begin{array}{l}\max\big\{ \ThetaL(1), \\
\ThetaL \big[(\frac{\gamma^2}{Kq^{M-1}})^{\frac{1}{M}} \big] \big\}\end{array} & \text{If } \limK  \frac{\gamma}{K}=0 \text{ and }\limK\frac{K}{q\gamma^2}>0  \\
\ThetaL (\frac{\gamma^2}{K}) & \text{If } \limK  \frac{\gamma}{K}=0
\text{ and } \limK\frac{K}{q\gamma^2} =0\\
\ThetaL \left[\left(\gamma/q^{M-1}\right)^{\frac{1}{M}} \right] & \text{If }\limK  \frac{\gamma}{K}>0 \text{ and } \limK\frac{1}{q\gamma}>0  \\
\ThetaL (\gamma) & \text{If }\limK \frac{\gamma}{K}>0 \text{ and }
\limK\frac{1}{q\gamma} =0
\end{array} \right.
\end{array}
\end{equation}
\end{Lemma}
\begin{proof}
Using Lemma \ref{lem:delay_relay_network} in Appendix C, the results
in Lemma \ref{lem:stability_base} follows using similar argument as
in Appendix E.
\end{proof}

The following Corollary summarizes the performance gain of the OBDWF
protocol in stability region.
\begin{Corollary}[Stability Region Comparison]
\label{cor:stability} Let $\lambda_s^*(\text{OBDWF})$ and
$\lambda_s^*(\text{DDF})$ be the maximum source arrival rate the
system can support and maintain queue stability using OBDWF and
conventional DDF protocol respectively. For sufficiently large $K$,
the gain on the stability region is given by:
\begin{equation}
\label{eq:lambda_cmp}\begin{array}{l}
\frac{\lambda_s^*(\text{OBDWF})}{\lambda_s^*(\text{DDF})} =
\\
\left\{
\begin{array}{ll}
\begin{array}{l}\max \big\{ \ThetaL(1), \\
\ThetaL \big[(\frac{\gamma^2}{Kq^{M-1}})^{\frac{1}{M}} \big] \big\} \end{array}
& \text{If } \limK  \frac{\gamma}{K}=0 \text{ and }\limK\frac{K}{q\gamma^2}>0 \\
\ThetaL (\frac{\gamma^2}{K}) & \text{If } \limK  \frac{\gamma}{K}=0
\text{ and } \limK\frac{K}{q\gamma^2} =0\\
\ThetaL \left(K(\gamma q)^{\frac{1-M}{M}}\right) & \text{If }\limK  \frac{\gamma}{K}>0 \text{ and } \limK\frac{1}{q\gamma}>0  \\
\ThetaL (K) & \text{If }\limK \frac{\gamma}{K}>0 \text{ and }
\limK\frac{1}{q\gamma} =0
\end{array} \right.
\end{array}
\end{equation}
\end{Corollary}
\begin{Rem}[Interpretation of the Corollary \ref{cor:stability}]
Note that when the packet size $\overline{N}$ is large such that
$\gamma
>\ThetaL(\psi)$, then the OBDWF protocol can obtain an order-wise
gain. For example, when $\gamma=\ThetaL(K)$ and $q=\ThetaL(1)$, then
$\frac{\lambda_s^*(\text{OBDWF})}{\lambda_s^*(\text{DDF})}=\ThetaL(K)$.
~\hfill\IEEEQED
\end{Rem}

\subsection{Delay Performance}
In this section, we shall compare the average end-to-end packet
delay performance. The average end-to-end packet delay of the relay
network is defined below.
\begin{Def}[Average End-to-End Packet Delay]\label{def:avg_delay}
Let $I_i^a$ and $I_i^r$ be the frame indices of the $i$-th packet
arrival at the source node and the $i$-th packet successfully
received at the destination node respectively. The average
end-to-end packet delay\footnote{Note that it is implicitly assumed
that the system of queues are stable when we discuss average delay
because otherwise, the probability measure behind the {\em
expectation} is not defined.} of the relay network is defined as
$\overline{D} =
\lim_{\kappa\rightarrow\infty}\frac{\sum_{i=1}^\kappa
I_i^r-I_i^a}{\kappa}$. ~\hfill \IEEEQED
\end{Def}

The following Theorem summarizes the average delay performance of
the proposed OBDWF protocol.
\begin{Theorem}{\em (Average End-to-End Packet Delay of the OBDWF Protocol):}\label{thm:avg_delay}
For sufficiently large $K$ and $\lambda_s$ in the stability region
in Theorem \ref{thm:stability}, the average end-to-end packet delay
of the OBDWF protocol is given by:
\begin{equation}
\label{eq:delay_obdmf} \overline{D}_{\text{OBDWF}}=\max\left\{
\ThetaL\left(\frac{\lambda_s^{(2)}\zeta}{\lambda_s(1-\lambda_s\zeta)}\right),\ThetaL(\overline{D}_S)\right\}
\end{equation}
where $\zeta$ is given by (\ref{eq:zeta}) and $\overline{D}_S$ is
given by (\ref{eq:DR_maintext}). ~\hfill\IEEEQED
\end{Theorem}

\begin{proof}
Please refer to Appendix F.
\end{proof}

\begin{Rem}[Interpretation of Theorem \ref{thm:avg_delay}]
The first term in the RHS of (\ref{eq:delay_obdmf}) is the average
waiting time (number of frames) that the packets stays in the source
buffer. It is affected by the source arrival model, i.e.,
$\lambda_s$ and $\lambda_s^{(2)}$. The second term $\overline{D}_S =
\ThetaL( \overline{D}_R)$, where $\overline{D}_R$ is the average
time the packet stays in the relay network. This factor depends on
both the packet size $\overline{N}$ and the mobility of the network
($q$). Fig. \ref{fig:delay_order} further illustrates the asymptotic
relationship between $\overline{D}_S$, $\gamma$ and $q$.
Specifically, the x-axis is $\frac{\log\gamma}{\log K}$, and the
y-axis is $\frac{\log\overline{D}_S}{\log K}$. Observe that
$\overline{D}_S$ is an increasing function of $\gamma$ and $1/q$.
~\hfill \IEEEQED
\end{Rem}

\begin{figure}
 \begin{center}
  \resizebox{9cm}{!}{\includegraphics{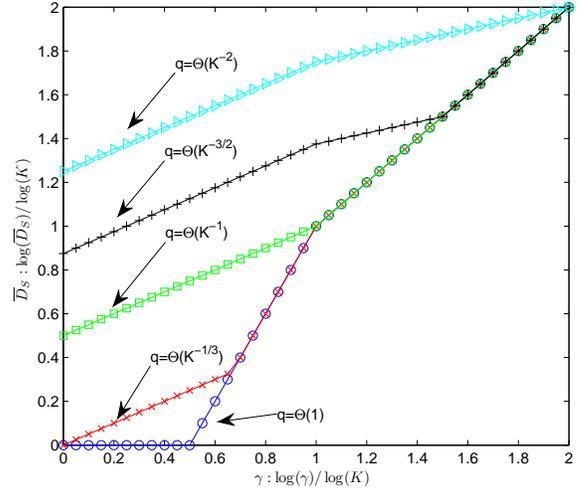}}
 \end{center}
    \caption{The asymptotic
        relationship between $\overline{D}_S$, $\gamma$ and $q$. Specifically, the
        x-axis is $\frac{\log(\gamma)}{\log(K)}$, and the y-axis is
        $\frac{\log(\overline{D}_S)}{\log(K)}$.}
    \label{fig:delay_order}
\end{figure}

Similarly, the delay performance of the conventional DDF protocol is
summarized in the following Lemma:
\begin{Lemma}{\em (Average End-to-End Packet Delay of the Conventional DDF Protocol):}\label{lem:avg_delay_base}
For sufficiently large $K$ and $\lambda_s$ in the stability region
in Lemma \ref{lem:stability_base}, the average end-to-end packet
delay under the conventional DDF protocol is given by:
\begin{equation}
\overline{D}_{\text{DDF}} =
\ThetaL\left(\frac{\lambda_s^{(2)}\overline{D}_S}{\lambda_s(1-\lambda_s\overline{D}_S)}\right)
\end{equation}
where $\overline{D}_S$ is given by (\ref{eq:DR_maintext}).
~\hfill\IEEEQED
\end{Lemma}
\begin{proof}
Please refer to the Appendix G.
\end{proof}

The following Corollary summarizes the average delay gain of the
OBDWF protocol.
\begin{Corollary}[Average Delay Comparison]
\label{cor:delay} For sufficiently large $K$ and $\lambda_s$ in the
stability region in Lemma \ref{lem:stability_base}, the average
delay gain of the OBDWF and the conventional DDF protocols is given
by:
\begin{equation}
\label{eq:delay_cmp}\begin{array}{l}
\frac{\overline{D}_{\text{DDF}}}{\overline{D}_{\text{OBDWF}}} =\\
\quad\quad\min\left\{\ThetaL\left(\frac{\lambda_s^{(2)}}{\lambda_s(1-\lambda_s\overline{D}_S)}\right),\ThetaL\left(\frac{\overline{D}_S(1-\lambda_s\zeta)}{\zeta(1-\lambda_s\overline{D}_S)}\right)
\right\}
\end{array}
\end{equation}
where $\zeta$ is given by (\ref{eq:zeta}) and
$\ThetaL(\overline{D}_S)$ is given by (\ref{eq:DR_maintext}).
~\hfill\IEEEQED
\end{Corollary}

\begin{Rem}[Interpretation of the Corollary
\ref{cor:delay}] There are several scenarios that the OBDWF protocol
will have significant order-wise gain on the delay performance. For
example, when $\lambda_s$ is close to the service rate
$1/\overline{D}_S$, i.e., $\lambda_s= \frac{1}{\overline{D}_S} -
\epsilon$ where $\epsilon = o(\frac{1}{\overline{D}_S})$, we have
$\frac{\overline{D}_{\text{DDF}}}{\overline{D}_{\text{OBDWF}}} =
\ThetaL(\frac{1}{\epsilon\overline{D}_S})>\ThetaL(1)$. On the other
hand, even if $\lambda_s$ is not so close to
$\frac{1}{\overline{D}_S}$, i.e.,
$1-\lambda_s\overline{D}_S=\ThetaL(1)$, there will still be
order-wise gain as long as $\frac{\lambda_s}{\lambda_s^{(2)}}= o(1)$
and $\frac{1}{\overline{D}_S}=o(1)$. Specifically, if
$\Pr\{A_s(t)=\ThetaL(K)\}=\nu$, $\Pr\{A_s(t)=0\}=1-\nu$, where
$\nu=\ThetaL(\frac{1}{K^2})$, $\gamma=\ThetaL(K)$ and
$q=\ThetaL(1)$, then we can obtain
$\frac{\overline{D}_{\text{DDF}}}{\overline{D}_{\text{OBDWF}}} =
\ThetaL(K)$ by (\ref{eq:delay_cmp}). ~\hfill\IEEEQED
\end{Rem}

\section{Simulation Results and Discussion}\label{sec:sim}

\begin{table*}
\begin{center}
\caption{Elaboration of the baselines. Specifically, baseline 1
refers to {\em conventional DDF} protocol\cite{DDF:2003}, baseline 2
refers to {\em conventional AF} protocol\cite{AF:sel}, baseline 3
refers to {\em AF with spatial combining} and baseline 4 refers to
{\em DF with Spatial Combining}.}
\begin{tabular}{|p{2.5cm}|p{12.5cm}|}
\hline Baseline Name                        & \hspace{5cm} Description \\
\hline Baseline 1

        (Conventional DDF) &
$\bullet$ The source node broadcasts a packet from the buffer at the
beginning of the frame until at least one relay or destination node
returns with an ACK.

$\bullet$ If the destination node returns with an ACK, the source
node start to broadcast a new packet; If the relay node returns with
an ACK, the source node stops broadcasting and the relay node
forward the packet to the destination node in the next frame.
\\
\hline  Baseline 2

(Conventional AF)      &

$\bullet$ The source node broadcasts a packet from the buffer at the
beginning of a frame.

$\bullet$ All the relays listen and store the received samples from
the source during the listening phase and the relay with the largest
metric ($\frac{P^2S_{s,j}S_{j,d}}{PS_{s,j}+PS_{j,d}+1}$, where
$S_{s,j}=\frac{H_{s,j}}{d_{s,j}^{\alpha}}$ and
$S_{j,d}=\frac{H_{j,d}}{d_{j,d}^{\alpha}}$ ) is selected to amplify
and forward to the destination node in the next frame.
\\
\hline Baseline 3

(AF with Spatial Combining) & $\bullet$ The source node broadcasts a
packet from the buffer at the beginning of a frame.

$\bullet$ All the relays listen and store the received samples from
the source during the listening phase and $N_A$ relays with the
largest metric ($\frac{P^2S_{s,j}S_{j,d}}{PS_{s,j}+PS_{j,d}+1}$,
where $S_{s,j}=\frac{H_{s,j}}{d_{s,j}^{\alpha}}$ and
$S_{j,d}=\frac{H_{j,d}}{d_{j,d}^{\alpha}}$ ) are selected to amplify
and forward to the destination node in the next frame.
\\
\hline Baseline 4

(DF with Spatial Combining) & $\bullet$ The source node broadcasts a
packet from the buffer at the beginning of the frame until at least
$N_D$ relay nodes or destination node return with an ACK.

$\bullet$ If the destination node returns with an ACK, the source
node starts to broadcast a new packet; If at least $N_D$ relay nodes
return with an ACK, the source node stops broadcasting and all the
relay nodes that have decoded the packet from the source node will
forward the packet to the destination node in the next frame.
\\
\hline
\end{tabular}
\label{tab:baselines}
\end{center}
\end{table*}

In this section, we shall compare the performance of the proposed
OBDWF protocol with various baseline schemes. Baseline 1 refers to
the {\em conventional DDF} protocol\cite{DDF:2003}, baseline 2
refers to the {\em conventional AF} protocol\cite{AF:sel}. {
Baseline 3 and 4 are extensions of Baseline 1 and 2 respectively
with spatial combining from multiple relays, which are elaborated in
Table \ref{tab:baselines}} We consider a system with the source node
at $(0,0)$ and the destination node at $(5,0)$. The $K$ relays are
randomly distributed between the source node and the destination
node, as shown in Fig. \ref{fig:system_model}. The movement of each
relay is given by the random walk mobility model in
(\ref{eq:random_walk}), where the number of relay mobility regions
is $M=5$. The small scale fading gain follows complex Guassian with
unit variance. The pass loss exponent $\alpha=4$, and the transmit
SNR $P=20$dB.  For bursty arrivals, we assume $\Pr\{A_s=15\} =
0.001$ and $\Pr\{A_s=0\} = 0.999$. This corresponds to an arrival
rate $\lambda_s = 0.015$. The packet size
$\overline{N}=W\tau\log_2K$, the frame duration $\tau = 5$ ms and
the bandwidth $W=1$MHz. Using Lemma \ref{lem:macro-base}, the
physical data rate at the source node is set to be $R=W\log_2K$,
which is the optimal choice for conventional DDF.

Fig. \ref{fig:thrput} illustrates the average end-to-end system
throughput $\overline{T}$ versus the number of relays $K$ at
different mobility $q$. Observe that the proposed OBDWF protocol has
significant gain compared with the baselines. Furthermore, the
performance of the OBDWF protocol is insensitive to the mobility of
the network $q$. Fig. \ref{fig:stability} illustrates the maximal
stable arrival rate $\lambda_s$ versus the number of relays $K$ at
different network mobility $q$ under the bursty source model.
Similar significant gains over the baselines can be observed. {
Moreover, it can be observed in these two figures that the
simulation results match with the theoretical results derived in
Section \ref{sec:thrput} and \ref{sec:delay}.}

{ Fig. \ref{fig:delay_bursty_finite} illustrates the average
end-to-end packet delay versus the number of relays $K$ at different
mobility $q$ with finite buffer length of 25 packets for all the
nodes. Note that, the delay performance is an increasing function of
$1/q$ for all protocols and there is also a significant gain of the
proposed OBDWF protocol. }

Fig. \ref{fig:delay_lambda} illustrates the average end-to-end
packet delay versus the average arrival rate $\lambda_s$ with
infinite buffer length for all the nodes. Note that the baseline 2
and 3 are not stable under the operating regime considered. The
delay performance of baseline 4 quickly blows up at $\lambda_s
=0.08$, which is the maximal stability input rate for this baseline.
On the other hand, the delay performance of the proposed OBDWF
protocol significantly out-performs all the baseline over the entire
range of traffic loading. { Fig. \ref{fig:thrput_new_mobility}
illustrates the average end-to-end system throughput $\overline{T}$
versus the number of relays $K$ under the {\em Random Waypoint
Model}, which is also widely used in
\cite{Mobility:1996,Mobility:1998,Mobility:1999}. Similar
performance gains can be observed.}

\begin{figure}
\begin{center}
  \subfigure[$q=1/10$ ]
  {\resizebox{8cm}{!}{\includegraphics{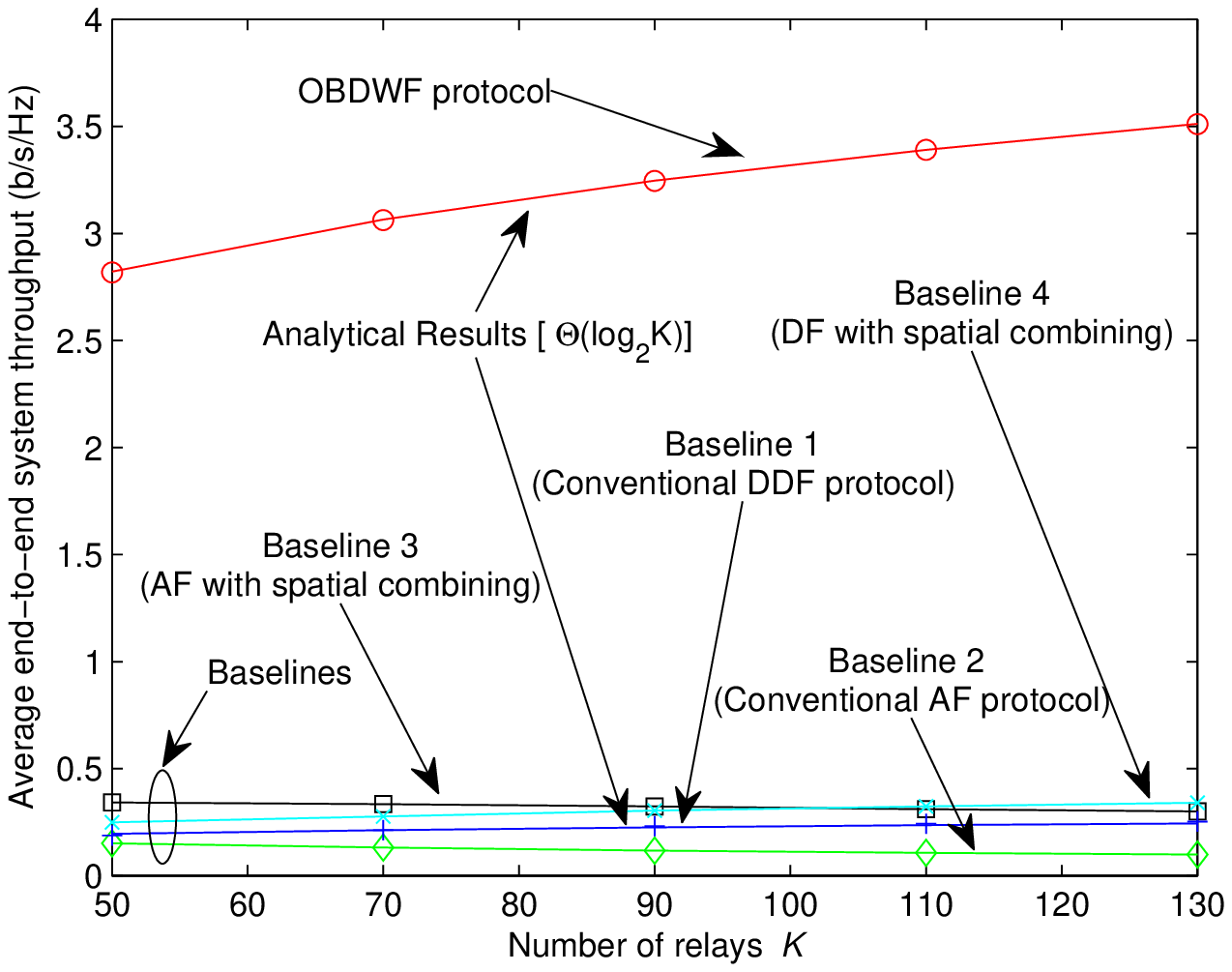}}}
  \subfigure[$q=1/5$ ]
  {\resizebox{8cm}{!}{\includegraphics{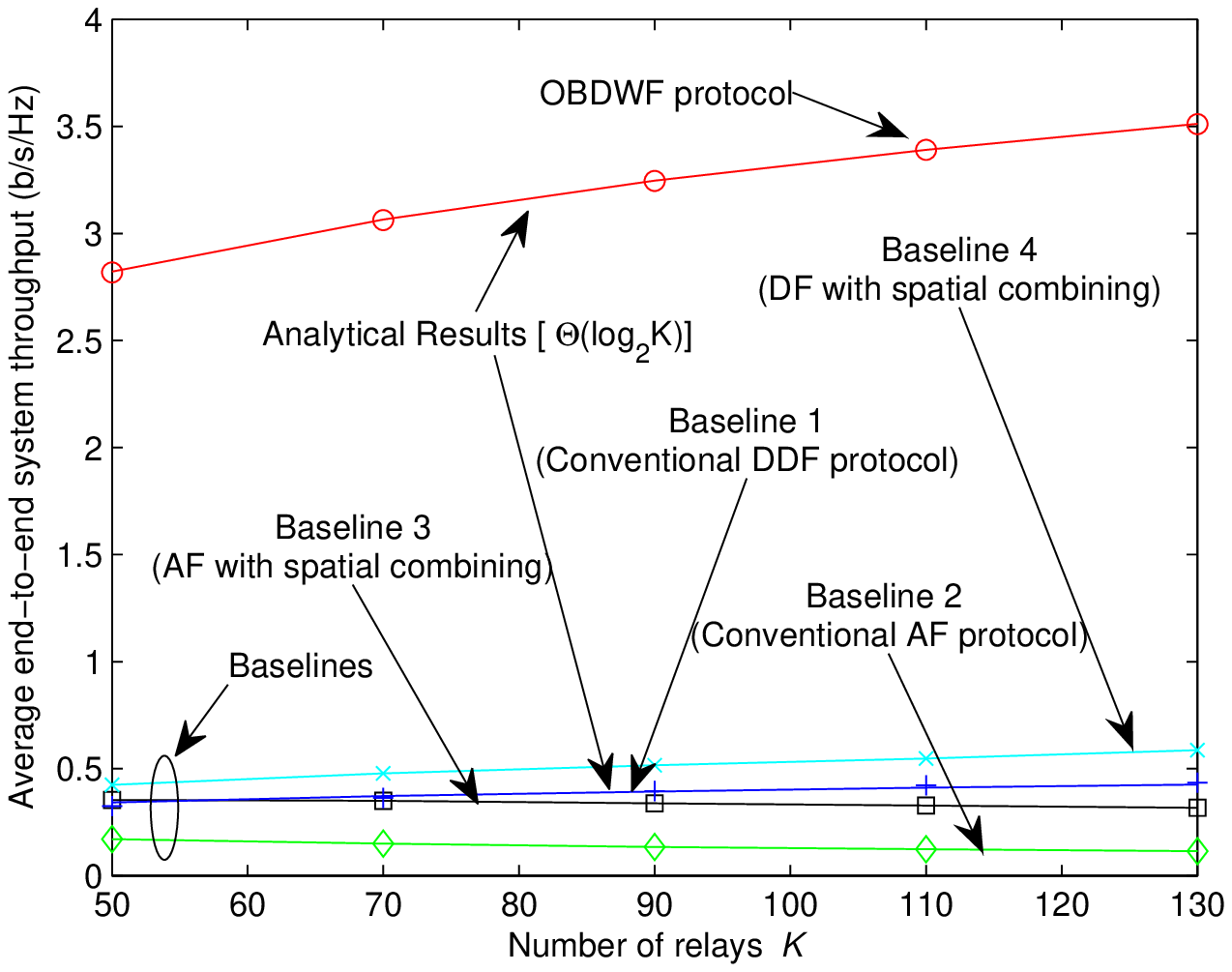}}}
  \end{center}
    \caption{ Average end-to-end system throughput $\overline{T}$ versus the number of relays $K$ at
different mobility $q$. The physical data rate is $R=W\log_2K$
according to Lemma \ref{lem:macro-base}, $(N_A,N_D) = (5,5)$ in the
baseline 3 and 4 respectively. The marks
``o,x,+,$\square,\diamond$'' denote simulation results and the solid
lines represent the analytical results for OBDWF and baseline 1
(conventional DDF) protocol.}
    \label{fig:thrput}
\end{figure}

\begin{figure}
\begin{center}
  \subfigure[$q=1/10$ ]
  {\resizebox{8cm}{!}{\includegraphics{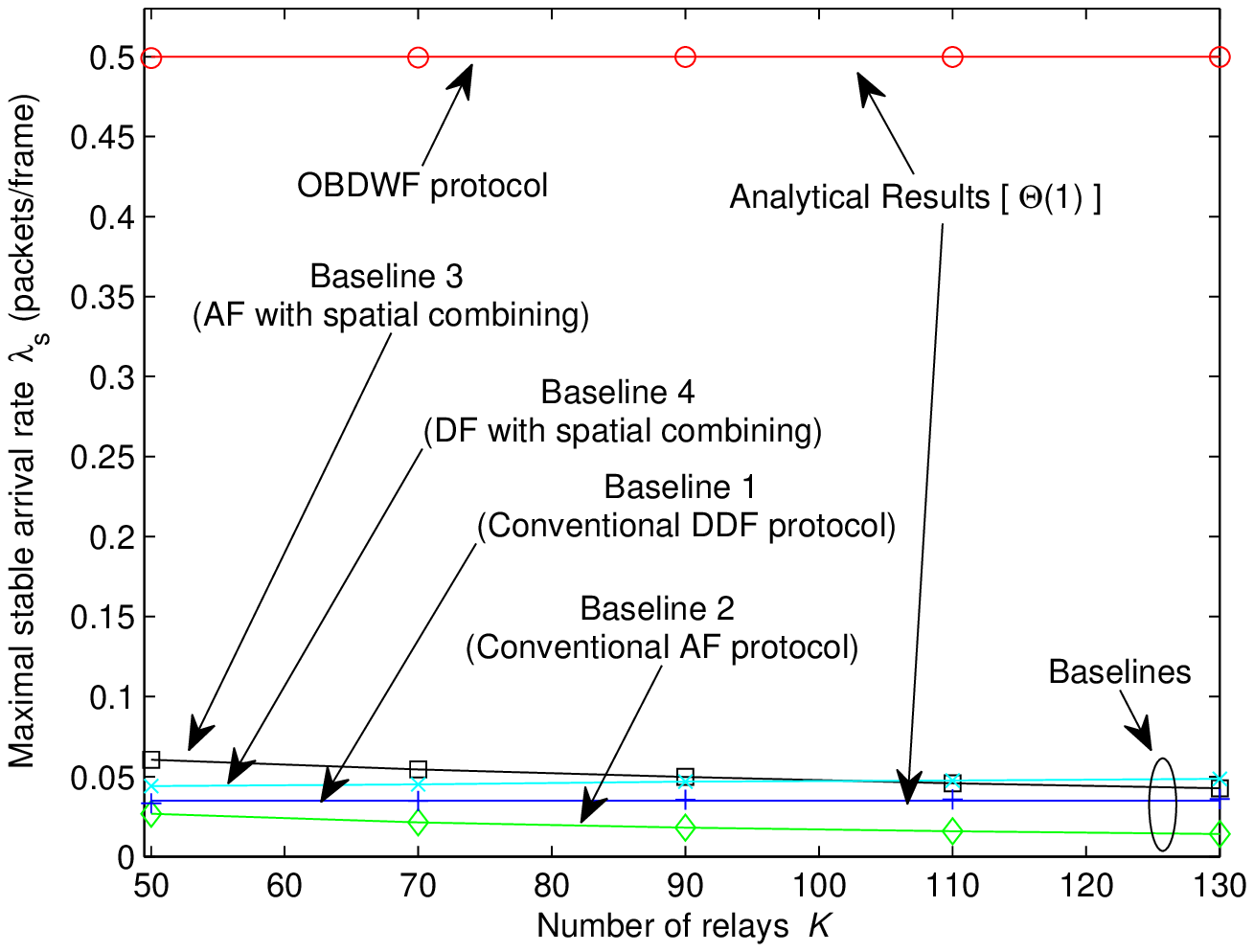}}}
  \subfigure[$q=1/5$ ]
  {\resizebox{8cm}{!}{\includegraphics{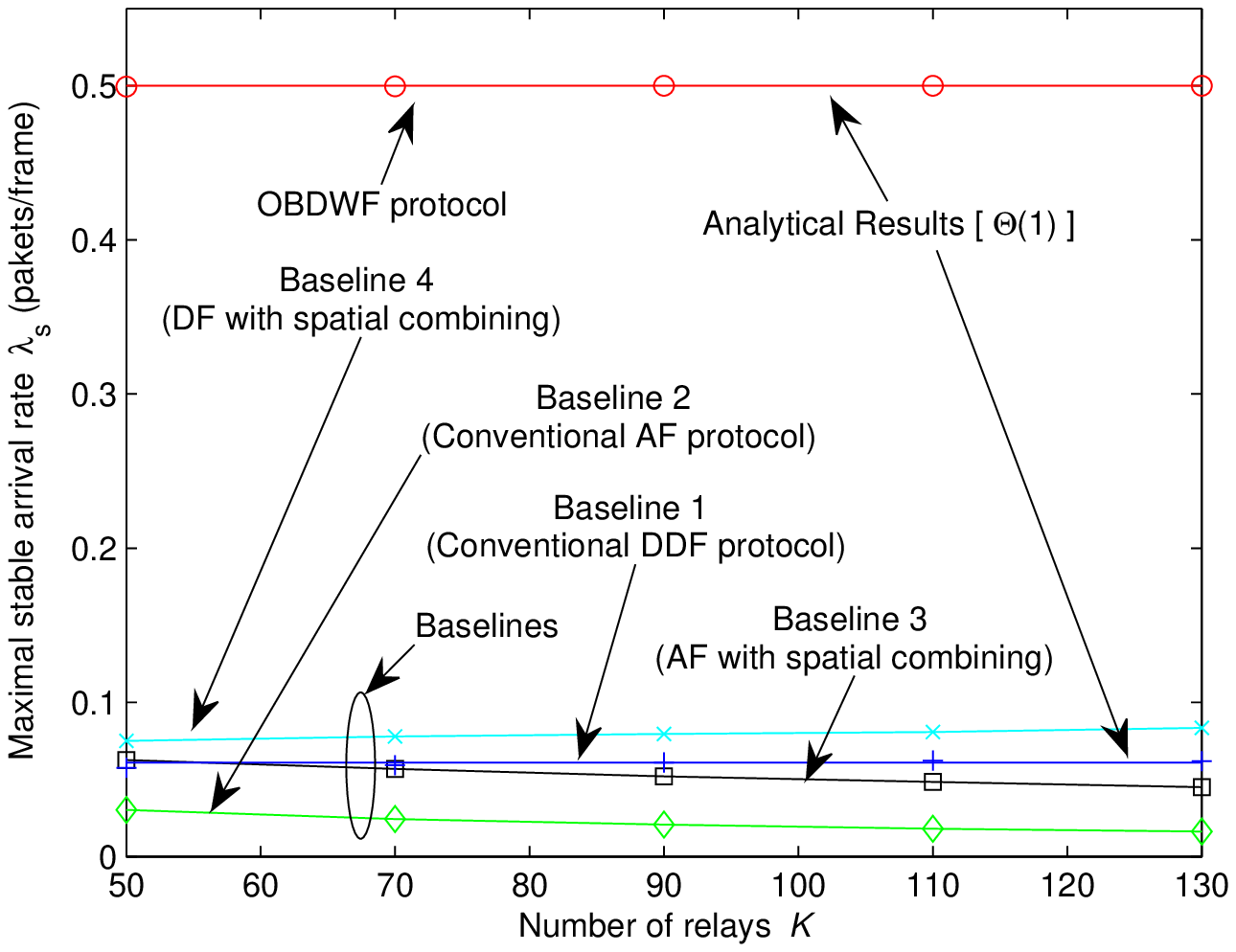}}}
  \end{center}
    \caption{ The maximal stable arrival rate $\lambda_s$ versus the number of
relays $K$ at different mobility $q$. The packet size is
$\overline{N}=W\tau\log_2K$, $(N_A,N_D) = (5,5)$ in the baseline 3
and 4 respectively. The marks ``o,x,+,$\square,\diamond$'' denote
simulation results and the solid lines represent the analytical
results for OBDWF and baseline 1 (conventional DDF) protocol. }
    \label{fig:stability}
\end{figure}

\begin{figure}
\begin{center}
  \subfigure[$q=1/10$ ]
  {\resizebox{8cm}{!}{\includegraphics{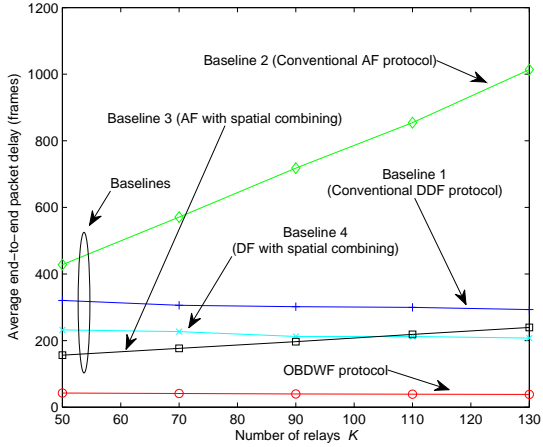}}}
  \subfigure[$q=1/5$ ]
  {\resizebox{8cm}{!}{\includegraphics{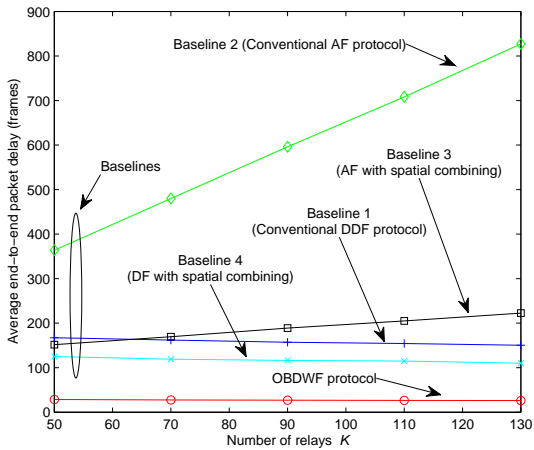}}}
 \end{center}
    \caption{
    The average end-to-end packet delay versus the number of
relays $K$ at different mobility $q$ with finite buffer length of 25
packets for all the nodes. The packet size is
$\overline{N}=W\tau\log_2K$, $(N_A,N_D) = (5,5)$ in the baseline 3
and 4 respectively. The source arrival model is given by
$\Pr\{A_s=15\} = 0.001$ and $\Pr\{A_s=0\} = 0.999$.
    }
    \label{fig:delay_bursty_finite}
\end{figure}

\begin{figure}
 \begin{center}
  \resizebox{8cm}{!}{\includegraphics{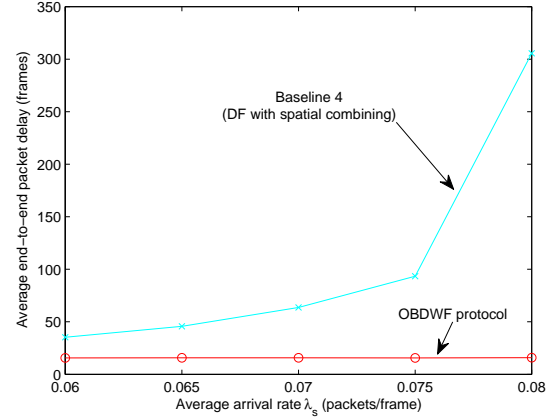}}
 \end{center}
    \caption{
    The average end-to-end
packet delay versus the average arrival rate $\lambda_s$. Note that
the other baselines are not stable under the operating regime
considered. The packet size is $\overline{N}=W\tau\log_2K$, $K=110$,
$q=1/5$ and $N_D=5$ in the baseline 4. The source arrival model is
given by $\Pr\{A_s=1\}=\lambda_s$, and $\Pr\{A_s=0\}=1-\lambda_s$.
    }
    \label{fig:delay_lambda}
\end{figure}

\begin{figure}
 \begin{center}
  \resizebox{8cm}{!}{\includegraphics{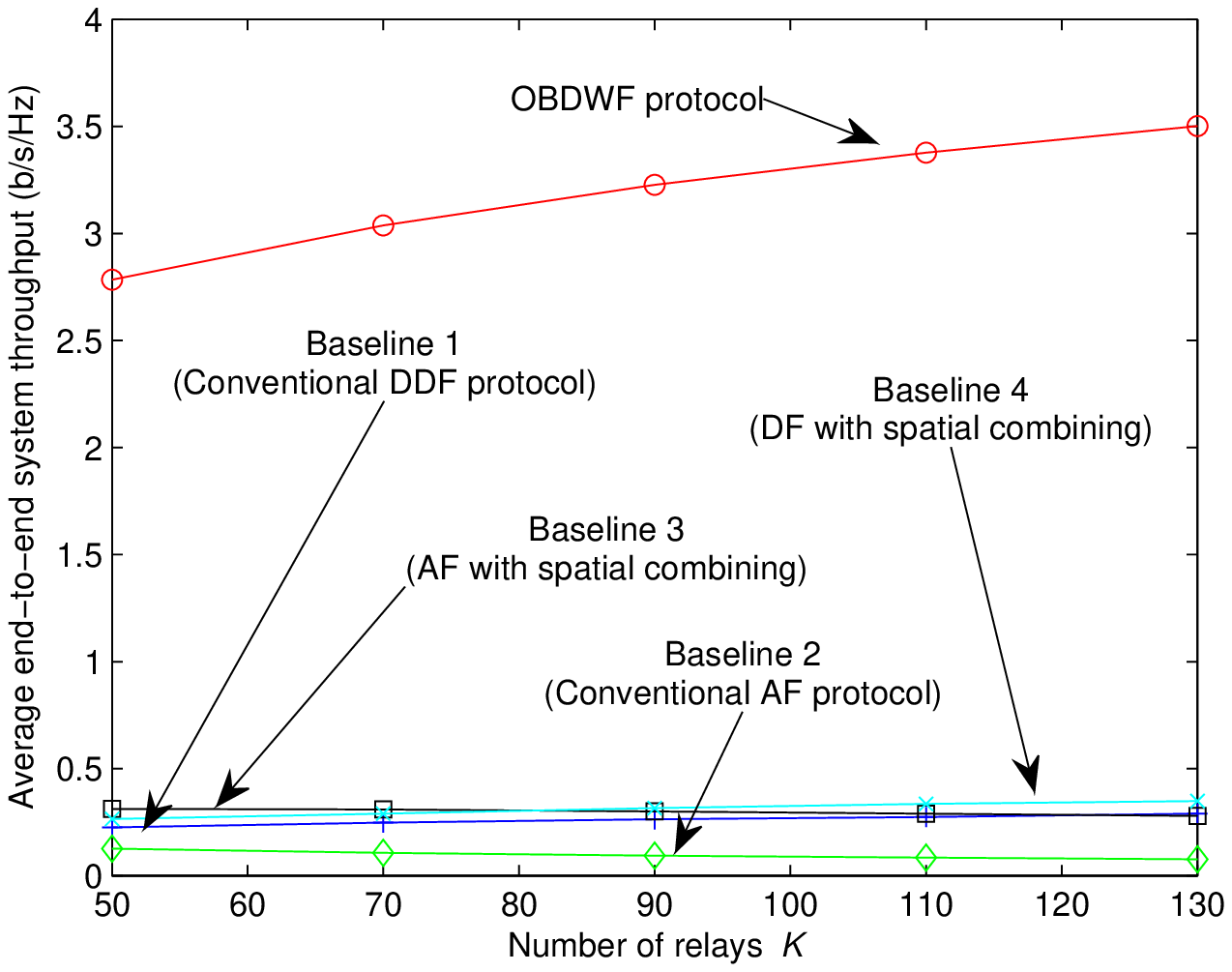}}
 \end{center}
    \caption{
    Average end-to-end system throughput $\overline{T}$ versus the number of relays $K$
    under the random waypoint model \cite{Mobility:1996,Mobility:1998,Mobility:1999}.
    Specifically, a node chooses a random destination distributed uniformly in the coverage area and moves in
    that direction with a random speed chosen uniformly in
    the interval [0.1, 0.6] (per frame). On reaching the destination the node
    pauses for some random time chosen uniformly in
    the set \{0,1,$\cdots$,10\} (frames) and the process repeats itself.  The
physical data rate is $R=W\log_2K$, $(N_A,N_D) = (5,5)$ in the
baseline 3 and 4 respectively.  }
    \label{fig:thrput_new_mobility}
\end{figure}

\section{Conclusions} \label{sec:con}
In this paper, we propose an opportunistic buffered
decode-wait-and-forward (OBDWF) protocol for a point-to-point
communication system with $K$ mobile relays. Unlike conventional DDF
protocol, the proposed OBDWF protocol exploits both the relay
buffering and relay mobility in the systems. We derive closed-form
expressions on the asymptotic system throughput under infinite
backlog as well as the average end-to-end delay under a general
bursty arrival model,  Based on the analysis, we found that there
exists a throughput delay tradeoff in the {\em buffered relay
network}. The system can achieve a higher throughput
$\ThetaL(\log_2K)$ using the proposed OBDWF protocol at the expense
of extra delay. The system mobility affects the tradeoff as below:
\begin{itemize}
\item{\bf Effect on the Throughput/Stability Region Performance:}
According to Theorem 1, the maximal average system throughput of the
proposed OBDWF protocol $\ThetaL(\log_2K)$ is not influenced by the
relay mobility.
\item{\bf Effect on the Delay Performance:}
If the movement of the mobility is fast (large region transition
probability $q$), the chance one relay with source's packet gets
close to the destination is high, leading to small delay in the
relay network, and vice versa. This can be observed from Theorem 3.
\end{itemize}

Finally, theoretical and numerical results demonstrate the
significant performance gains of the proposed OBDWF protocol against
various baseline references.

\appendices

\section*{Appendix A: The Proof of Lemma \ref{lem:relay_decode}}
\label{app:relay_decode} Without loss of generality, we only study
the number of connected links with the source. The connection with
the destination can follow the same approach.


For a given small scale fading realization $H_{s,j}$, the coverage
radius of the source within which the relays have connected links to
the source is given by $d=(P\xi
H_{s,j})^{\frac{1}{\alpha}}/(\beta-1)^{\frac{1}{\alpha}}$.
Therefore, the area of the source's coverage can be approximated as
$\frac{\pi d^2}{2}$, and the probability one relay falls into the
source's coverage is $\frac{d^2}{2R^2}$. Taking consideration of the
random realization of $H_{s,j}$, the probability that one relay
falls into the source's coverage is given by
\begin{equation}
\phi=\int_{0}^{\infty}\frac{d^2}{2R^2}f(H_{s,j})\text{d}H_{s,j}=\ThetaL(\frac{1}{\gamma})
\end{equation}
where $f(H_{s,j})$ is the pdf of the small scale fading gain
$H_{s,j}$.

Suppose there are $X$ relays having connected links with the source,
therefore, $\Pr(X=x)= \binom{K}{x}
\left(\phi\right)^x\left(1-\phi\right)^{K-x}$. $X$ can be treated as
one binomial random variable generated from Bernoulli trail where it
has value 1 with probability $\phi$. In other words, $ X \sim
Binomial(K,\phi)$.

When $\limK \frac{\gamma}{K}=0$, i.e., $\limK \phi K = \infty$, by
law of large number, we have
\begin{equation}
\Pr\left[\phi -o(\phi )<X/K<\phi +o(\phi )\right] \rightarrow 1
\quad \mbox{when} \quad K \rightarrow \infty,
\end{equation}
where $\limK\frac{o(\phi )}{\phi }=0$. Therefore, the $ \Pr\{X
=\ThetaL(\phi K)=\ThetaL(\frac{K}{\gamma}) \}$ tends to 1 when
$K\rightarrow \infty$. This finishes the proof of statement I.

When $\limK \frac{\gamma}{K}>0$, since the probability that one
relay has the connected link with source node is
$\phi=\ThetaL(\frac{1}{\gamma})$, the probability that at least one
relay can receive the transmitted packet is given by $1-(1-\phi)^{K}
= \ThetaL(K/\gamma)$. This finishes the proof of statement II.

\section*{Appendix B: The Proof of Theorem
\ref{thm:avg_thrput}}\label{app:thm_thrput}

{ When $\gamma=\ThetaL(1)$, it is obvious that
$\overline{T}_{\text{OBDWF}}=\ThetaL(1)$, and hence
$\gamma=\ThetaL(1)$ is not the optimal choice. When $\limK
\frac{\gamma}{K}=0$ and $\gamma>\ThetaL(1)$, the coverage radius of
the source and destination nodes can be made sufficiently small
given a fixed transmit power $P$. Only the relays close to the
source node (in region 1) can decode the packet transmitted from the
source node and only the relays close to the destination (in region
$M$) can forward the packet to the destination node.} Furthermore,
there are (with probability $1$) some relays in region 1 can decode
the source's packets by Lemma \ref{lem:relay_decode}. Similarly,
there are (with probability 1) some relays in region $M$ can forward
packets to the destination, and hence the average system throughput
is given by: $\overline{T}_{\text{OBDWF}} \doteq \frac{R}{2} =
\frac{W}{2}\log_2\beta$, where $\doteq$ denotes the equality with
high probability (with probability $1-\frac{1}{K}$ as in
\cite{GamalAE:04}).

When increasing the $\beta$ such that $\limK \frac{\gamma}{K}>0$,
The source should continue to transmit $\ThetaL(\frac{\gamma}{K})$
slots before $\ThetaL(1)$ relays can decode the packet by Lemma
\ref{lem:relay_decode}, and hence the average system throughput is
$\overline{T}_{\text{OBDWF}} =\ThetaL \big( \frac{K\log_2
\beta}{\gamma} \big)$. Since the log function is in the numerator,
increasing the order of $\beta$ will reduce the order of
$\frac{K\log_2 \beta}{\gamma}$. As a result, the optimal $\gamma$ is
$\gamma^*=O(K)$. $\gamma^*$ should be infinitely close to
$\ThetaL(K)$ from below to make that there are always relays in the
connected region of the source/destination node ($\limK
\frac{\gamma}{K}=0$). The corresponding maximal average system
throughput is infinitely close to $\frac{W\alpha}{4} \log_2 K$, when
$K\rightarrow\infty$.

\section*{Appendix C: The Proof of Lemma
\ref{lem:macro-base}}\label{app:macro-base} Obviously, under the
conventional DDF protocol, the average system throughput for a given
data rate $ R = W \log_2 \beta $ is given by
\begin{equation}
\overline{T}_{\text{DDF}} = \ThetaL (\log_2
\beta/\overline{D}_S)=\ThetaL (\log_2 \gamma/\overline{D}_S)
\end{equation}
where $ \overline{D}_S $ is the average service time (number of
frames) for a packet, i.e., the average time spent for the source
node to transmit a packet to the destination node. $\overline{D}_S$
can be divided into two part, i.e.,
$\overline{D}_S=\overline{\rho}+\overline{\eta}$, where
$\overline{\rho}$ is the average time spent for the source node to
transmit a packet to the relay network, and $\overline{\eta}$ is the
average time spent for the relay network to forward the said packet
to the destination. Specifically, we have following Lemma for the
average service time $\overline{D}_S$.
\begin{Lemma}{\em (Average Service Time of Conventional DDF Protocol): }\label{lem:delay_relay_network}
The average service time for a packet under conventional DDF
protocol is given by: $\overline{D}_S = \overline{\rho} +
\overline{\eta} = \ThetaL(\overline{\eta})$, where $\overline{\rho}
= \max\left\{ \ThetaL(1),\ThetaL(\frac{\gamma}{K}) \right\}$, and
\begin{equation}
\label{eq:delay_RS}\begin{array}{l} \overline{\eta} = \\
\left\{
\begin{array}{ll}
\begin{array}{l}\max \big\{ \ThetaL(1), \\
\ThetaL \big[(\frac{\gamma^2}{Kq^{M-1}})^{\frac{1}{M}} \big]
\big\}\end{array}
& \text{If } \limK  \frac{\gamma}{K}=0 \text{ and }\limK\frac{K}{q\gamma^2}>0  \\
\ThetaL (\frac{\gamma^2}{K}) & \text{If } \limK  \frac{\gamma}{K}=0
\text{ and } \limK\frac{K}{q\gamma^2} = 0\\
\ThetaL \left[\left(\gamma/q^{M-1}\right)^{\frac{1}{M}} \right] & \text{If }\limK  \frac{\gamma}{K}>0 \text{ and } \limK\frac{1}{q\gamma}>0 \\
\ThetaL (\gamma) & \text{If }\limK \frac{\gamma}{K}>0 \text{ and }
\limK\frac{1}{q\gamma} = 0
\end{array} \right.
\end{array}
\end{equation}
\end{Lemma}
\begin{proof}
please refer to Appendix D.
\end{proof}

Note that $ \overline{D}_S $ is an increase function of $\gamma$ by
Lemma \ref{lem:delay_relay_network}. Since
$\overline{T}_{\text{DDF}}= \ThetaL (\frac{\log_2
\gamma}{\overline{D}_S})$ and due to the $\log$ function in the
numerator, we have $\overline{D}_S(\gamma^*)=\ThetaL(1)$, where
$\gamma^*$ is the optimal value that maximizes
$\overline{T}_{\text{DDF}}$, i.e.,
$\gamma^*=\arg\max_{\gamma}\overline{T}_{\text{DDF}}$. According to
the delay expression in (\ref{eq:delay_RS}), when $\gamma^* = \max
\{\ThetaL(1), \ThetaL(\sqrt{Kq^{M-1}})\}$, we have
$\overline{D}_S(\gamma^*)=\ThetaL(1)$. As a result, if
$\sqrt{Kq^{M-1}}=o(1) $, the optimal value is $ \gamma^*=
\ThetaL(1)$, leading to $ \overline{T}_{\text{DDF}}^* =
\ThetaL(Kq^{M-1})^{1/M}$ . Otherwise, the optimal value is $
\gamma^*= \ThetaL(\sqrt{Kq^{M-1}})$, leading to $
\overline{T}_{\text{DDF}}^* =
 \ThetaL \left(\log_2
(Kq^{M-1}) \right)$.

\section*{Appendix D: The proof of Lemma
\ref{lem:delay_relay_network}}\label{app:delay_relay_network} We
provide the proof in two scenarios, $\limK \frac{\gamma}{K}=0$ and
$\limK \frac{\gamma}{K}>0$, respectively.

\subsection*{A. $\limK \frac{\gamma}{K}=0$ Scenario}
When the source broadcasts a packet, there are
$\mathbf{\Theta}(\frac{K}{\gamma})$ relays can decode this packet
with probability 1 by Lemma \ref{lem:relay_decode}, and hence
$\overline{\rho}=\ThetaL(1)$. The movement of relays with the said
packet can be divided in to two stage: (1) un-balanced, which means
the order of relays with the said packet in each region is not the
same; (2) balanced, which means the order of relays with the said
packet in each region is the same. Obviously, after
$\mathbf{\Theta}(1/q)$ frames, the system is balanced, i.e., there
are $\mathbf{\Theta}(\frac{K}{\gamma})$ relays with the said packet
in each of the $M$ regions.


When the connected link mainly happens in the un-balanced stage,
after $\overline{\eta}=O(1/q)$ frames, the number of relays with the
said packet in the region $ M $ is $N_M= \frac{K}{\gamma}
(q\overline{\eta})^{M-1} $. The chance that these relays have
connected link with the destination is $ \min \left\{ 1,
\ThetaL(\frac{N_M}{\gamma})\right\} $. It can be obtained following
the same approach as Lemma \ref{lem:relay_decode}. Therefore, after
$ W(\overline{\eta}) = \max \left\{ 1,
\ThetaL(\frac{\gamma}{N_M})\right\} $ frames, there is at least one
relay with the said packet having connected link to the destination
node. Increasing the order of $\overline{\eta} $, the number of
relays with the said packet in the region $ M $ increases, but $
W(\overline{\eta}) $ decrease. The actual delay should satisfy:
$\overline{\eta} = \ThetaL\left( W(\overline{\eta}) \right)$, which
leads to
\begin{equation}
\label{eq:eta} \overline{\eta} = \max \left\{ \ThetaL(1), \ThetaL
\left[(\frac{\gamma^2}{Kq^{M-1}})^{1/M} \right] \right\}
\end{equation}
As a result, the requirement of $\overline{\eta}=O(1/q)$ is
satisfied when $ \limK\frac{K}{q\gamma^2}>0 $.

When the connected link mainly happens in the balanced stage, thus
$\overline{\eta}>\ThetaL(1/q)$, i.e., $\limK
q\overline{\eta}=\infty$. This requirement is satisfied when $
\limK\frac{K}{q\gamma^2}=0 $. In this case, the average delay is
mainly due to the waiting time after the relays' movement is
balanced, given by $ \overline{\eta} = \ThetaL (\frac{\gamma^2}{K})
$.

\subsection*{B. $\limK \frac{\gamma}{K}>0$ Scenario}
When $\limK \frac{\gamma}{K}>0$,  the source should continue to
transmit $\ThetaL(\frac{\gamma}{K})$ slots before $\ThetaL(1)$ relay
can decode the packet by Lemma \ref{lem:relay_decode}, i.e.,
$\overline{\rho}=\ThetaL(\frac{\gamma}{K})$. Therefore, there are
$\ThetaL(1)$ relays with the packet in the relay network rather than
$\ThetaL(\frac{K}{\gamma})$ as the $\limK \frac{\gamma}{K}=0$
scenario. Following the similar approach as in the above subsection
by replacing $\ThetaL(\frac{K}{\gamma})$ with $\ThetaL(1)$, it can
be shown that the average delay is given by Lemma
\ref{lem:delay_relay_network}.

\section*{Appendix E: The Proof of the Theorem
\ref{thm:stability}}\label{app:stability} In this proof, we shall
first study the stability region of the source buffer $Q_s(t)$, and
then prove that under the same stability region, the relay buffers
$\{Q_1,\cdots,Q_K\}$ are stable too.

From \cite{Stability:1999,Stability:2005}, the queueing system is
stable, if and only if the average arrival rate $\lambda_s$ is
smaller than the service rate $1/b$, i.e., $\lambda_s b<1$, where
$b$ is the average service time to server a queueing packet out of
the queueing buffer. For the source buffer $Q_s$, the average
service time is the average number of frames for the source node to
transmit a packet to the relay network, denoted as $\zeta$. From
Lemma \ref{lem:relay_decode}, we can discuss in the following two
scenarios.

If $\limK\frac{\gamma}{K}=0$, when the source broadcasts a packet,
there are $\mathbf{\Theta}(\frac{K}{\gamma})$ relays can decode this
packet with probability 1. Therefore, $\zeta=\ThetaL(1)$.

If $\limK\frac{\gamma}{K}>0$, the source should continue to transmit
$\ThetaL(\frac{\gamma}{K})$ slots before $\ThetaL(1)$ relays can
decode the packet. Therefore, $\zeta=\ThetaL(\frac{\gamma}{K})$.

Note that the $K$ queues $\{Q_1(t),\cdots,Q_K(t)\}$ are
statistically identical, and they are either all stable or all
unstable. Consider a fictitious queueing system with
$Q_r(t)=\sum_{k=1}^{K}Q_k(t)$ with the average arrival rate
$\lambda_r$ and service rate $1/b_r$. Obviously, $Q_k(t) \leq
Q_r(t),\forall k$ and hence, all the relay queues
$\{Q_1(t),...,Q_k(t)\}$ are dominated by the fictitious queue
$Q_r(t)$. The average arrival rate of the fictitious queue is
$\lambda_r=\min\{\lambda_s,\frac{1}{\zeta}\}$, and the average
number of frames for the fictitious system to deliver a packet to
the destination $b_r<\zeta$. Therefore, if
$\lambda_s<\frac{1}{\zeta}$, then $\lambda_r<\frac{1}{b_r}$. In
other words, the queues of the relay nodes are also stable if the
queue $Q_s(t)$ is stable.

\section*{Appendix F: The Proof of the Theorem
\ref{thm:avg_delay}}\label{app:avg_delay} Note that
$\overline{D}_{\text{OBDWF}}=\overline{D}_Q+\overline{D}_R$, where
$\overline{D}_Q$ is the average queueing delay in the source buffer
before transmitted to the relay network, and $\overline{D}_R$ is the
average waiting time in the relay network. Following the same
approach as the proof of Lemma \ref{lem:delay_relay_network}, it is
easy to verify that $\overline{D}_R=\ThetaL(\overline{\eta})$ given
in (\ref{eq:delay_RS}). The remaining task is to find
$\overline{D}_Q$, which is discussed below.

Let $X$ be the number of frames (namely service time) to transmit a
packet into the relay network, and denote $b=\mathbb{E}[X]$ and
$b^{(2)}=\mathbb{E}[X^2]$ respectively. Furthermore, let $n_i$
denotes the number of packets in the source buffer $Q_s$ immediately
after transmitting the $i$-th packet to the relay network. $a_i$ is
the number of packets arriving during the service time of the $i$-th
packet, $\overline{a}$ is the number of packets arriving in one
frame given that $\overline{a}\geq1$, i.e.,
$\Pr\{\overline{a}\}=\Pr\{A_s=\overline{a}|\overline{a}\geq1\}$, and
$\widetilde{a}_{i+1}$ is the number of packets arriving during the
service time of the $(i+1)$-th packet minus one frame (The number of
arrivals during the last $m-1$ frames if the service time of the
$(i+1)$-th packet is $m$ frames). Then $n_i$ will form a Markov
chain with the following transitions.
\begin{equation}
\label{eq:n_i} n_{i+1}=\left\{\begin{array}{ll}
\overline{a}-1+\widetilde{a}_{i+1}
&\text{if } n_i = 0 \\
n_i-1+a_{i+1} &\text{otherwise}
\end{array}\right.
\end{equation}
Specifically, the probability generating function (p.g.f)
$\overline{A}(z)$ of $\overline{a}$ is given by
\begin{equation}
\label{eq:A_z}\begin{array}{lll}
\overline{A}(z)&=&\sum\nolimits_{\overline{a}=1}^{\infty}\Pr\{A_s=\overline{a}|\overline{a}\geq1\}z^{\overline{a}}\\
&=&
\sum\nolimits_{\overline{a}=1}^{\infty}\frac{\Pr\{A_s=\overline{a}\}z^{\overline{a}}}{1-P_0}=\frac{\Lambda(z)-P_0}{1-P_0}
\end{array}
\end{equation}
where $P_0$ is the probability that no packets arrives.
$\Lambda(z)=\sum_{A_s=0}^{\infty}\Pr\{A_s\}z^{A_s}$ is the p.g.f of
the bursty arrival $A_s$. The p.g.f $A(z)$ of the number $a$
arriving within a service time is
\begin{eqnarray}
\label{eq:A_z_1}\begin{array}{lll}
A(z)&=&\sum_{a=0}^{\infty}\Pr\{a\}z^a\\
&=&\sum_{l=1}^{\infty}\Pr\{X=l\}\sum_{a=0}^{\infty}\Pr\{a|X=l\}z^a\\
&=&\sum_{l=1}^{\infty}\Pr\{X=l\}[\Lambda(z)]^l
\end{array}
\end{eqnarray}
Similarly, the p.g.f $\widetilde{A}(z)$ of packet arrival
$\widetilde{a}_i$ is given by
\begin{equation}
\label{eq:A_z_2}
\widetilde{A}(z)=\sum\nolimits_{l=1}^{\infty}\Pr\{X=l\}[\Lambda(z)]^{l-1}=B(\Lambda(z))/\Lambda(z)
\end{equation}
where $B(z)=\sum_{l=1}^{\infty}\Pr\{X=l\}z^l$ is the p.g.f of the
service time. From (\ref{eq:n_i})-(\ref{eq:A_z_2}), the p.g.f
$P_n(z)$ of the number in the system immediately after the service
completion instants is\cite{Bose:2002}
\begin{equation}
P_n(z)=\frac{(1-\lambda_s
b)[1-\Lambda(z)]B(\Lambda(z))}{\lambda_s\Lambda(z)[B(\Lambda(z))-z]}
\end{equation}

It is shown in \cite{Bruneel:1993} that the p.g.f $P(z)$ for the
packets in the queueing system immediately after an arbitrary frame
is given by
\begin{equation}
\begin{array}{lll}
P(z)=P_n(z)\frac{\lambda_s(1-z)\Lambda(z)}{1-\Lambda(z)}=\frac{(1-\lambda_s
b)(1-z)B(\Lambda(z))}{[B(\Lambda(z))-z]}
\end{array}
\end{equation}
Therefore, by Little's law\cite{Bose:2002}, the average time a
packet spends in the buffer will be
\begin{equation}
\label{eq:delay_expression}
\overline{D}_Q=\frac{\lambda_s^2b^{(2)}-\lambda_s^2b-\lambda_s
b+\lambda_s^{(2)}b}{2\lambda_s(1-\lambda_s b)}+b
\end{equation}

Note that, to make the system stable, the arrival rate $\lambda_s$
should be smaller than the service rate $1/b$, i.e., $\lambda_s
b<1$. It agrees with the stability condition given in
\cite{Stability:1999,Stability:2005}.

If $\limK\frac{\gamma}{K}=0$, when the source broadcasts a packet,
there are $\mathbf{\Theta}(\frac{K}{\gamma})$ relays can decode this
packet with probability 1 by Lemma \ref{lem:relay_decode}.
Therefore, $b=\ThetaL(1)$ and $b^{(2)}=\ThetaL(1)$.

If $\limK\frac{\gamma}{K}>0$, the first and second order moments of
the service time for a packet to enter the relay network is
$b=\ThetaL(\frac{\gamma}{K})$ and
$b^{(2)}=\ThetaL(\frac{\gamma^2}{K^2})$ respectively from Lemma
\ref{lem:relay_decode}.

Let $\zeta=\max\left\{\ThetaL(1),\ThetaL(\frac{\gamma}{K})\right\}$,
and note that $\frac{\lambda_s}{\lambda_s^{(2)}}=O(1)$. The average
end-to-end packet delay is given by
\begin{eqnarray}\begin{array}{lll}
&&\overline{D}_{\text{OBDWF}}=\overline{D}_Q+\overline{D}_R\\
&=&\max\left\{
\ThetaL\left(\frac{\lambda_s^{(2)}\zeta}{\lambda_s(1-\lambda_s\zeta)}\right),
\ThetaL(\overline{D}_S)\right\}
\end{array}
\end{eqnarray}

\section*{Appendix G: The Proof of The Lemma
\ref{lem:avg_delay_base}}\label{app:avg_delay_base} From Appendix C,
the average service time to server a packet to the destination node
from the source node is
$b=\overline{D}_S=\mathbb{E}[D_S]=\ThetaL(\overline{D}_S)$. In the
followings, we shall prove that the seconder order moment of service
time is given by $b^{(2)}=\ThetaL\left((\overline{D}_S)^2\right)$.
Specifically,
\begin{eqnarray}\begin{array}{lll}
b^{(2)}&=&\mathbb{E}[D_S^{2}]=\sum_{D_S}\Pr\{D_S\}D_S^{2}\\
&=&\ThetaL\left[(\overline{D}_S)^2+f^-(K)+f^+(K)\right]
\end{array}
\end{eqnarray}
where $f^{-}(K)$ (in terms of $K$) is contributed by the
possibilities that the service time $D_S=o(\overline{D}_S) $, and
$f^{+}(K)$ is contributed by the possibilities that the service time
$D_S > \ThetaL(\overline{D}_S)$. Clearly, $ f^{-}(K) $ is
neglectable as $f^{-}(K)= o\left((\overline{D}_S)^2\right) $, i.e,
$b^{(2)}=\ThetaL\left[(\overline{D}_S)^2+f^+(K)\right]$.

We first consider the scenario where $\limK\frac{\gamma}{K}=0$ and
$\limK\frac{K}{q\gamma^2}<\infty$. In this case, the average service
time is given by $b=\overline{D}_S=\ThetaL(\frac{\gamma^2}{K})$, and
after $\ThetaL(\frac{\gamma^2}{K})$ frames, there are
$\ThetaL(\frac{K}{\gamma})$ relays with the transmitted packet in
the region $M$ from Lemma \ref{lem:delay_relay_network}. The
probability that one relay with the packet has connected link to the
destination is $\ThetaL(\frac{K}{\gamma^2})$. Denote $g(K)$ as the
service time that is order-wisely larger than $\frac{\gamma^2}{K}$,
i.e., $g(K)>\ThetaL(\frac{\gamma^2}{K})$ and the contribution of
$g(K)$ in $b^{(2)}$ is given by
\begin{eqnarray}\begin{array}{lll}
&\Pr\{D_S=g(K)\}g^2(K)\\
=&\ThetaL\left((1-K/\gamma^2)^{g(K)}g^2(K)\right)
=o\left((\gamma^2/K)^2\right)
\end{array}
\end{eqnarray}
As a result,
$b^{(2)}=\ThetaL\left((\gamma^2/K)^2\right)=\ThetaL\left((\overline{D}_S)^2\right)$.

In other scenarios, we can follow the same steps. The only
difference is that after $\overline{D}_S$ frames, the average number
of relays with the transmitted packet in the region $M$  is not
$\ThetaL(\frac{K}{\gamma})$ anymore. Specifically, it is given in
the proof of Lemma \ref{lem:delay_relay_network}.

Given $b=\ThetaL(\overline{D}_S)$ and
$b^{(2)}=\ThetaL\left((\overline{D}_S)^2\right)$, following the same
steps as in Appendix F, the average end-to-end packet delay for the
conventional DDF protocol is given by
\begin{eqnarray}\begin{array}{l}
\overline{D}_{\text{DDF}}=\max\left\{\ThetaL\left(\frac{(\lambda_s^{(2)}/\lambda_s-1)\overline{D}_S}{1-\lambda_s\overline{D}_S}\right),
\ThetaL\left(\frac{\overline{D}_S}{1-\lambda_s\overline{D}_S}\right)\right\}\\
=\ThetaL\left(\frac{\lambda_s^{(2)}\overline{D}_S}{\lambda_s(1-\lambda_s\overline{D}_S)}\right)
\end{array}
\end{eqnarray}

\bibliographystyle{IEEEtran}
\bibliography{IEEEabrv,ray,huang}

\end{document}